\documentclass[3p]{elsarticle}

\usepackage[utf8]{inputenc}
\usepackage{bm}
\usepackage{algorithm}
\usepackage[noend]{algpseudocode}
\usepackage{booktabs}
\usepackage{graphicx} 

\usepackage{mathtools}

\usepackage[english]{babel}

\usepackage{amssymb,amsmath,amsthm} 

\usepackage{mathrsfs}

\usepackage{bbm}

\newtheorem{theorem}{Theorem}

\newtheorem{corollary}[theorem]{Corollary}
\newtheorem{lemma}[theorem]{Lemma}
\newtheorem{definition}[theorem]{Definition}
\newtheorem{conjecture}[theorem]{Conjecture}

\algnewcommand\algorithmicforeach{\textbf{for each}}
\algdef{S}[FOR]{ForEach}[1]{\algorithmicforeach\ #1\ \algorithmicdo}

\begin{document}

\title{Lower Bounds for the Minimum Spanning Tree Cycle Intersection Problem}

\author[1,2]{Manuel Dubinsky\corref{cor1}}
\ead{mdubinsky@undav.edu.ar}

\author[3,4,5]{Kun-Mao Chao\fnref{fn1}}

\author[6]{César Massri}

\author[7]{Gabriel Taubin}

\cortext[cor1]{Corresponding author}
\fntext[fn1]{Kun-Mao Chao was supported in part by the Ministry of Science and Technology, Taiwan, under grant number 111-2221-E-002-132-MY3.}

\address[1]{Ingeniería en Informática, Departamento de Tecnología y Administración, Universidad Nacional de Avellaneda, Argentina}
\address[2]{Departamento de Computación, Facultad de Ciencias Exactas y Naturales, Universidad de Buenos Aires, Argentina}  
\address[3]{Department of Computer Science and Information Engineering, National Taiwan University, Taipei, Taiwan}
\address[4]{Graduate Institute of Biomedical Electronics and Bioinformatics, National Taiwan University, Taipei, Taiwan}
\address[5]{Graduate Institute of Networking and Multimedia, National Taiwan University, Taipei, Taiwan}
\address[6]{Departamento de Matemática, Universidad de CAECE, CABA, Argentina}
\address[7]{School of Engineering, Brown University, Providence, RI, USA}

\begin{abstract}
Minimum spanning trees are important tools in the analysis and design of networks. Many practical applications require their computation, ranging from biology and linguistics to economy and telecommunications. The set of cycles of a network has a vector space structure. Given a spanning tree, the set of non-tree edges defines cycles that determine a basis. The intersection of two such cycles is the number of edges they have in common and the intersection number -denoted $\cap(G)$-, is the number of non-empty pairwise intersections of the cycles of the basis. The Minimum Spanning Tree Cycle Intersection problem consists in finding a spanning tree such that the intersection number is minimum. This problem is relevant in order to integrate discrete differential forms. In this paper, we present two lower bounds of the intersection number of an arbitrary connected graph $G=(V,E)$. In the first part, we prove the following statement:

$$\frac{1}{2}\left(\frac{\nu^2}{n-1} - \nu\right) \leq \cap(G),$$

where $n = |V|$ and $\nu$ is the \emph{cyclomatic number} of $G$. In the second part, based on some experimental results and a new observation, we conjecture the following improved tight lower bound:

$$(n-1) \binom{q}{2} + q \ r\leq \cap(G),$$

where $2 \nu = q (n-1) + r$ is the integer 
division of $2 \nu$ and $n-1$. This is the first result in a general context, that is for an arbitrary connected graph.
\end{abstract}	


\maketitle

\section{Introduction}

Spanning trees play an important role in algorithm design~\cite{wu2004spanning}. The Minimum Spanning Tree Cycle Intersection problem was proposed recently~\cite{Dubinsky2021} as part of a -still under development- mesh deformation algorithm in the context of the field of Digital Geometry Processing~\cite{Botsch:2010}. More specifically, a \emph{mesh} is a discretization of a surface represented as a graph embedded in 3D Euclidean space. A \emph{deformation} is a consistent perturbation of the vertices of the mesh to achieve some visual effect, which is relevant in simulations and 3D animation.

\medskip

The proposed deformation algorithm is based on interactively setting the desired edge lengths globally or locally -in a region of the mesh-, and automatically rearrange vertex positions in the best possible way to 
approximate those lengths. Concretely, the situation can be expressed in these terms: let $G = (V,E)$ be a directed connected graph and $w:  E \rightarrow \mathbb{R}$ be an edge function. We call $w$ a \emph{discrete 1-form} on $G$. Integrating $w$ is the problem of finding a vertex function $x: V \rightarrow \mathbb{R}$ such that for every edge $e_{vw}=(v \rightarrow w)$, the differential $dx:  E \rightarrow \mathbb{R}$, defined as: $dx(e_{vw}) := x(w) - x(v)$, approximates $w$ minimizing the error:

$$\lambda(x) = \sum_{e_{vw} \in  E} \left[ dx(e_{vw}) - w(e_{vw}) \right] ^2.$$

The directed incidence matrix $D \in \{0,1,-1\}^{|E| \times |V|}$ of $G$ models $dx$ in the sense that each row contains the information of an edge $e_{v,w}$: -1 in the column of $v$, 1 in the column of $w$ and 0 otherwise. $D$ has rank $n-1$, and its kernel is generated by $\bm 1$ -the vector of all 1's. This degree of freedom can be explained by the fact that each translation of the function $x$ (i.e.: $x + k$) has the same differential.

\medskip

The set of cycles of a graph has a vector space structure over $\mathbb{Z}_2$, in the case of undirected graphs, and over $\mathbb{Q}$, in the case of directed graphs \cite{Kavitha:2009}. A basis of such a vector space is denoted \emph{cycle basis} and its dimension is the \emph{cyclomatic number} $\nu = |E| - |V| + |CC|$ where $E$, $V$ and $CC$ are the set of edges, vertices and connected components of the graph, resp. Given a cycle basis $B$ we can define its \emph{cycle matrix} $\Gamma \in K^{|E| \times \nu}$ -where $K$ is the scalar field (i.e.: $\mathbb{Z}_2$ or $\mathbb{Q}$)-, as the matrix that has the cycles of $B$ as columns. 

\medskip

Let $G$ be a graph and $T$ a spanning tree of $G$. For an edge $e$ in $G-T$, there is a cycle in $T \cup {e}$. We call those edges \emph{cycle-edges} and those cycles \emph{tree-cycles}. The set of tree-cycles of $T$ defines a cycle basis of $G$.

\medskip

The initial version of the mesh deformation algorithm consisted in choosing a spanning tree $T$ of $G$ and solving the following  linear system:

$$
M
\begin{bmatrix}
	\bf \hat x\\
	\bf y
\end{bmatrix}
=
\bf w,
$$

\noindent where $\bf \hat x \in \mathbb{R}^{\big|V \big|-1}$ is the solution of the integration problem -fixing the value of and arbitrary vertex, in order to eliminate the degree of freedom introduced by $dx$-, $\bf y \in \mathbb{R}^{\nu}$ measures an error associated to each tree-cycle, $\bf w\in \mathbb{R}^{\big|E\big|}$ is the evaluation of $w$ for every edge, and $M$ is a matrix of this form: 

$$
M = \left[
\begin{array}{c|c}
\hat{D} & \Gamma
\end{array}\right],
$$

\noindent where $\hat D \in \{0,1,-1\}^{|E| \times |V|-1}$ is the directed incidence matrix -without its first column-, and $\Gamma \in \{0,1,-1\}^{|E| \times \nu}$ is the cycle matrix induced by $T$. As the columns of $\hat D$ and $\Gamma$ generate orthogonal subspaces, $M$ has full rank and consquently the linear system has a unique solution.

\medskip

A natural question to consider is: how to choose $T$ such that its induced cycle matrix $\Gamma$ is as sparse as possible, in order to apply fast linear solvers. In more general terms, this is the formulation of the  Minimum Cycle Basis problem, defined independently by 
Stepanec \cite{Stepanec:1964} and Zykov \cite{Zykov:1969}, that can be expressed as: 
given a connected graph, find a cycle basis such that the sum of the length of its cycles is minimum.
The variant of this problem where the basis is induced by the cycles of a spanning tree is the Minimum Fundamental Cycle Basis (MFCB) defined by Hubicka and Syslo \cite{Hubicka:1975}. A complete description of this family of problems is developed by Kavitha et al.~\cite{Kavitha:2009}. The computational complexity of the MFCB problem belongs to the NP-Hard class \cite{Deo:1982}. Given this computational restriction, we considered an equivalent version of the linear system:

$$
M^t M
\begin{bmatrix}
	\bf \hat x\\
	\bf y
\end{bmatrix}
=
M^t \bf w.
$$

\noindent where the \emph{Gram} matrix $M^t M$ has the following form:
$$
M^t M
=
\left[
\begin{array}{c|c}
\hat L & 0 \\
0 & \hat \Gamma
\end{array}
\right]
$$

\noindent where $\hat L$ is the 
\emph{Laplacian matrix} of $G$ -without its first column-, and 
$\hat \Gamma$ is the \emph{cycle intersection matrix} induced by $T$, that is: a matrix in which the $ij$-entry contains the number of edges that the tree-cycles $c_i$ and $c_j$ have in common. This matrix is interesting because its blocks have good properties: there are fast solvers for Laplacian matrices and $\hat \Gamma$ is -at least- symmetric, and it might have other good properties.

\medskip

In this context an equivalent natural question can be considered: how to choose $T$ such that $\hat \Gamma$ is as sparse as possible. 
This motivated us to formulate the Minimum Spanning Tree Cycle Intersection (MSTCI), which can be expressed as: given a connected graph $G$, find a spanning tree $T$ such that the number of non-empty pairwise intersections of its tree-cycles is minimum. 

\medskip

The MSTCI problem is not well studied yet, for example its complexity class is unknown. One interesting result states that if $G$ contains a \emph{universal vertex} -one that is connected to every other vertex-, then the \emph{star spanning tree} is a solution~\cite{CHEN202219}; in particular this can solve the case of dense graphs: $|E| > (|V|^2-2|V|)/2$. 

\medskip

For a spanning tree $T$ of a connected graph $G$, the complement of the sparsity of $\hat \Gamma$ can be measured by the intersection number w.r.t. $T$, $\cap_T(G)$: the number of non-empty pairwise intersections between its tree-cycles. Expressed in these terms, a solution of the MSTCI problem minimizes this quantity which we call the intersection number of $G$, denoted $\cap(G)$.

\medskip

In this paper, we present two lower bounds of $\cap(G)$ of an arbitrary connected graph $G$. This is relevant to measure the sparsity of $\hat \Gamma$, and it is the first result in a general context.

\medskip

The structure of the article is as follows. Section 2 sets notation, convenient definitions and auxiliary results. Section 3 presents a proof of a lower bound of the intersection number of a connected graph: 

$$\frac{1}{2}\left(\frac{\nu^2}{n-1} - \nu\right) \leq \cap(G),$$

\noindent where $n = |V|$. In Section 4, an improved tight lower bound is conjectured based on experimental results: 

$$(n-1) \binom{q}{2} + q \ r\leq \cap(G),$$

\noindent where $2 \nu = q (n-1) + r$ is the integer division of $2 \nu$ and $n-1$.~Finally, Section 5 summarizes the conclusions. 

\section{Preliminaries}


\subsection{Notation}
Let $G=(V,E)$ be a graph and $T$ a spanning tree of $G$. 
The number of vertices and edges will be  
$|V| = n$ and $|E| = m$, resp. The number of 
cycles of a cycle 
basis is $\nu = m - n + 1$, known as the 
\emph{cyclomatic number} of $G$. The 
unique path between $u,v \in V$ in the spanning tree $T$ is denoted $uTv$. The degree of a vertex $v \in V$ will be denoted $d(v)$.

\medskip

We will 
refer to the edges $e\in T$ as 
\emph{tree-edges} and to the edges 
in $G-T$ as \emph{cycle-edges}.
Every cycle-edge $e$ induces a cycle in $T \cup 
\{e\}$, which we will call a \emph{tree-cycle}. 
We will denote $\cap_G(T)$ to 
the number of non-empty tree-cycle pairwise 
intersections w.r.t. $T$, being $\cap(G) = 
\cap_G(T)$ the \emph{intersection number} 
of $G$ in the case where $T$ is a solution of the 
MSTCI problem. 

\medskip

A \emph{universal vertex} $u$ of $G$ is a vertex incident to 
all the other vertices, i.e. $d(u) = n - 1$. We shall call \emph{star spanning tree} to one that has one vertex that connects to all other vertices -note that a universal vertex induces a star spanning tree.$G$ is a \emph{regular graph} if all its vertices have the same degree. 

\medskip

Finally, we will refer to the terms \emph{vertex} 
and \emph{node} interchangeably.

\subsection{Auxiliary lemmas and definitions}

In this section we present some auxiliary lemmas 
and definitions. In the following, let $G=(V,E)$ be a 
connected graph and $T = (V, E')$ a spanning tree of $G$.

\medskip

\begin{definition}~\cite{Bondy2008}
Let $e \in E'$ be a tree-edge, and consider the two 
connected components $T_1$ and $T_2$ determined 
by $T - \{e\}$. We define the 
\emph{bond} $b_e$  
as the maximal set of edges $(v,w) \in E$ such 
that $v \in T_1$ and $w \in T_2$. 
\end{definition} 

\begin{lemma}\label{lemma:bond_pairwise_intersect}
Let $e \in E'$ be an edge of $T$ and $b_e$ its 
corresponding bond. Then, the edges in $b_e - \{e\}$ are 
cycle-edges that determine tree-cycles that intersect pairwise.
\end{lemma}
\begin{proof}
Let $T_1, T_2$ be the two connected 
components determined by $T-\{e\}$, and let 
$(v,w) \in b_e-\{e\}$. By definition, $v \in T_1$ 
and $w \in T_2$, thus $vTw$ must contain $e$ and 
$vTw \cup (v,w)$ determines a tree-cycle. This 
implies that $(v,w)$ is a cycle-edge and that the 
edges in $b_e-\{e\}$ intersect pairwise.
\end{proof}

\begin{lemma}\label{lemma:cycle_edge_every_bond}
Let $f = (v,w) \in E-E'$ be a cycle-edge. Then $f$ 
belongs to the bonds corresponding to the edges of 
$vTw$.
\end{lemma}
\begin{proof}
Let $e \in vTw$. Consider the two connected 
components $T_1, T_2$ determined by $T - \{e\}$. 
Clearly $v \in T_1$ and $w \in T_2$ then by 
definition $f \in b_e$. 
\end{proof}

\begin{corollary}\label{coro:cycle_edge_2_bonds}
Every cycle-edge belongs to at least two bonds.
\end{corollary}
\begin{proof}
Let $f = (v,w) \in E-E'$ be a cycle-edge. Since every cycle contains at least three edges, it follows that  
$vTw$ has at least length two. By the previous 
lemma the claim follows.
\end{proof}

Let $B_T = \{b_e\}_{e \in E'}$ be the set of bonds 
w.r.t. $T$. Then -by Lemma~\ref{lemma:bond_pairwise_intersect}- 
computing the total number of pairs of cycle-edges 
in $B_T$ overestimates $\cap_G(T)$:

$$\sum_{e \in E'} \binom{|b_e| - 1}{2} \geq 
\cap_G(T).$$

The overestimation is due to redundant 
intersections. Note that if the intersection of two tree-cycles $c_1, c_2$ contains two or more tree-edges, then the corresponding pair of cycle-edges will be contained in more than one bond. The next definition addresses this problem at the cost of incurring in underestimation.

\begin{definition} 
Let $B_T = \{b_e\}_{e \in E'}$ be 
the set of bonds w.r.t. $T$. Then a \emph{non-redundant bond set} is a set $\hat B_T = 
\{\hat b_e\}_{e \in E'}$ such that:
\begin{itemize}
	\item $\hat b_e \subseteq b_e$
	\item every edge $e \in E$ belongs to exactly one $\hat b_e$.
\end{itemize}
\end{definition}

Note that as $E$ is finite, the bonds of $T$ and 
$B_T$ exist and are finite. Thus in order to define a \emph{non-redundant bond set} $\hat B_T$, it suffices to remove all but one occurrence of the duplicated edges in the bonds of $B_T$. 

\begin{lemma}\label{lemma:bound}
If $T$ is a solution for the MSTCI problem 
w.r.t. $G$, $\hat B_T = \{\hat b_e\}_{e \in E'}$ a non-redundant bond set and $\phi_e = |\hat b_e| - 1$ for 
every $e \in E'$, then 

\begin{itemize}
	\item $\sum_{e \in E'} \binom{\phi_e}{2} \leq \cap(G)$
	\item $\nu = \sum_{e \in E'} \phi_e$.
\end{itemize}
\end{lemma}
\begin{proof}
By Lemma~\ref{lemma:bond_pairwise_intersect} the edges in 
$\hat b_e - \{e\}$ are cycle-edges that intersect pairwise. 
There are $\phi_e =|\hat b_e| - 1$ 
cycle-edges in $\hat b_e - \{e\}$, then $\hat b_e$ 
accounts for $\binom{\phi_e}{2}$ pairwise 
intersections. Since every edge belongs to exactly 
one \emph{non-redundant bond}, every tree-cycle 
pairwise intersection is counted at most 
once, then the first inequality holds. As each bond 
contains exactly $\phi_e$ cycle-edges and every 
cycle-edge belongs to exactly one bond, the second 
equation follows.
\end{proof}

The next Theorem clarifies the case of graphs that 
contain a universal vertex in the context of the MSTCI problem.

\begin{theorem}\label{teo:universal_vertex}~\cite{CHEN202219}
	If a graph $G$ admits a star spanning tree $T_s$, then

	$$\cap_G(T_s) \leq \cap_G(T)$$

	for any spanning tree $T$ of $G$.
\end{theorem}
	
\begin{corollary}\label{coro:star_solves_mstci}~\cite{CHEN202219}
	If a graph $G$ admits a star spanning tree $T_s$, then $T_s$
	is a solution of the MSTCI problem w.r.t. $G$.
\end{corollary}

\begin{lemma}\label{lemma:star_formula}~\cite{Dubinsky2021} 
	If $G$ is a graph that admits a star spanning $T_s$, then 
	$$\cap(G) = \cap_G(T_s) = \sum_{u \in V-\{v\}} \binom{d(u) - 1}{2}.$$
\end{lemma}

The next lemma is a well-known statement in convex optimization.

\begin{lemma}\label{lemma:strictly_convex}\cite[\S 3.1.4]{Boyd2004-mk}
	A quadratic function $f: \mathbb{R}^n \rightarrow 
	\mathbb{R}$ defined as,
	
	$$f(x_1, \dots, x_n) = \bm x^T \bm Q \bm x + \bm c^T \bm x + d$$
	is strictly convex if and only if $Q$ is positive definite.
\end{lemma}
	
\begin{lemma}\label{lemma:matrix_is_strictly_convex}
	The matrix $M_n=\bm I + \bm 
	1$ 
	where $\bm I$ is the 
	$n \times n$ identity matrix and 
	$\bm 1$ is the $n \times n$ matrix 
	whose entries are all equal to 1 is positive 
	definite.
\end{lemma}
\begin{proof}	
	As $\bm 1$ has rank 1, then 0 is an eigenvalue with multiplicity $n-1$, and the remaining eigenvalue is $tr(\bm 1)=n$, the trace of $\bm 1$. So $\bm 1$ is positive semidefinite. As $\bm I$ is positive definite, we conclude that $M_n=\bm I + \bm 1$ is positive definite.
\end{proof}
	
\section{Proof of a lower bound}

A lower bound of the intersection number seems interesting 
both in theoretical and practical aspects. It gives an insight in the 
MSTCI problem, it can be useful for comparing algorithms and as mentioned in the introduction, it can approximate the sparsity of cycle intersection matrices. 
In this section we present a proof of a lower bound.

\bigskip

\begin{theorem}\label{teo:main}
If $G=(V,E)$ is a connected graph, then 
$$\frac{1}{2}\left(\frac{\nu^2}{n-1} - \nu\right) \leq \cap(G).$$
\end{theorem}
\begin{proof}
Let $T=(V,E')$ be a solution of the MSTCI 
problem w.r.t. $G$, and $\hat B_T = 
\{\hat b_e\}_{e \in E'}$. 
a \emph{non-redundant bond set}. By Lemma~\ref{lemma:bound}, 
the following equations hold:

\begin{enumerate}
\item $\sum_{e \in E'} \binom{\phi_e}{2} \leq \cap(G)$
\item $\nu = \sum_{e \in E'} \phi_e$,
\end{enumerate}
where $\phi_e = |\hat b_e| - 1$ for every $e \in E'$. 
The expression $\sum_{e \in E'} 
\binom{\phi_e}{2}$ can be identified with a 
point in the image of a function $f: 
\mathbb{R}^{n-1} \rightarrow \mathbb{R}$, defined as

$$f(x_1, \dots, x_{n-1}) = \sum_{i = 1}^{n-1} \binom{x_i}{2} = \sum_{i = 1}^{n-1} \frac{x_i(x_i-1)}{2}.$$

In this setting, in order to effectively calculate a 
lower bound of the intersection number of $G$, the 
above expressions lead to the following quadratic 
optimization problem~\cite[\S 9]{murty:2010},
\begin{align*}
\text{minimize} & \quad \sum_{i = 1}^{n-1} \binom{x_i}{2}\\
\text{subject to} &\quad  \nu = \sum_{i = 1}^{n-1} x_i.
\end{align*}
This minimization problem can be solved by restricting a 
degree of freedom of the objective function, defining 
$\hat f: \mathbb{R}^{n-2} \rightarrow \mathbb{R}$ as
$$\hat f(x_1, \dots, x_{n-2}) = f(x_1, \dots, x_{n-2}, \nu - \sum_{i=1}^{n-2} x_i).$$

The next step is to show that $\hat f$ is a strictly 
convex function and consequently has a unique global 
minimum. The function $\hat f$ can be expressed in matrix 
form as follows,
$$\hat f(x_1, \dots, x_{n-2}) = \frac{1}{2} 
\left( \bm x^T 
(\bm I + \bm 1^{n-2 \times n-2}) \bm x - 2 \nu \bm{x}^T 
\bm 1^{n-2} + \nu^2 - \nu
\right),
$$
where $\bm x = (x_1, \dots, x_{n-2})^T$, $\bm I$ is the 
$(n-2) \times (n-2)$ identity matrix, $\bm 1^{n-2 \times n-2}$ 
and  $\bm 1^{n-2}$ are the $(n-2) \times (n-2)$  
matrix and the $n-2$ vector respectively, 
whose entries are all equal to 1.
By Lemma~\ref{lemma:strictly_convex}, $\hat f$ is strictly 
convex if and only if $(\bm I + \bm 1^{n-2 \times n-2})$ is positive 
definite; this fact is proved in Lemma~\ref{lemma:matrix_is_strictly_convex}. To calculate the 
unique minimum of $\hat f$ suffices to check where 
its gradient vanishes,
$$\nabla \hat f = (\bm I + \bm 1^{n-2 \times n-2}) \bm x - \nu \bm 1^{n-2} = 0.$$
The solution of this linear system is
$$x_i = \frac{\nu}{n-1}, \quad \forall \ 1 \leq i \leq n-2$$
and 
$$x_{n-1} = \nu - \sum_{i=1}^{n-2} x_i = \frac{\nu}{n-1}.$$
The global minimum is
\[
\hat f\left(\frac{\nu}{n-1}, \dots, \frac{\nu}{n-1}\right) = 
\frac{1}{2}\left(\frac{\nu^2}{n-1} - \nu\right).
\]

\noindent Summarizing, the following inequalities hold,
$$\frac{1}{2}\left(\frac{\nu^2}{n-1} - \nu\right) \leq \sum_{e \in E'} \binom{\phi_e}{2} \leq \cap(G),$$
and the global minimum of $\hat f$ is a lower 
bound of the
intersection number of $G$ as claimed. 
\end{proof}

\bigskip

We will refer to the lower bound as

$$l_{n,m} := \frac{1}{2} \left(\frac{\nu^2}{n-1} - \nu\right).$$

\subsection{Evaluation}

The next lemma focuses on the quality of $l_{n,m}$ in the sense that quantifies its underestimation in a particular family of graphs. 

\begin{lemma}\label{lemma:bound_regular_graph}
If $G=(V,E)$ is a graph that contains a universal 
vertex $u \in V$ such that the subgraph $G - \{u\}$ 
is $k$-regular with $k \geq 4$, then 
$$\frac{1}{8} \leq \frac{l_{n,m}}{\cap(G)} \leq \frac{1}{4}.$$
\end{lemma}
\begin{proof}
As $d(v) = k$ $\forall v \in 
V-\{u\}$, and the star spanning tree is a solution 
of the MSTCI problem then according to the 
intersection number formula,
$$\cap(G) = \sum_{v \in V-\{u\}} \binom{d(u) - 1}{2} =
\sum_{v \in V-\{u\}} \binom{k - 1}{2} = \frac{(n-1) (k-1)(k-2)}{2}.$$
Taking into account that,
$$m = \frac{1}{2} \sum_{v \in V} d(v) = 
\frac{1}{2}\left((n-1) + (n-1) k\right) = \frac{(n-1)(k+1)}{2}.$$
Then,
$$\nu = m - (n - 1) = \frac{(n-1)(k+1)}{2} - (n - 1) 
= \frac{(n-1)(k-1)}{2}.$$

\noindent Next, we analyze $l_{n,m}$,
$$l_{n,m} = \frac{1}{2} \left(\frac{\nu^2}{n-1} - \nu\right) = 
\frac{(n-1)(k-1)(k-3)}{8},$$
and finally we express the quotient,
$$\frac{\frac{1}{2} (\frac{\nu^2}{n-1} - 
\nu)}{\cap(G)} = \frac{k-3}{4(k-2)}.$$

This implies that when $k=4$ and 
$k \rightarrow \infty$, the lower and upper bounds 
are met.
\end{proof}
\bigskip

This last result expresses two facts: 1) the 
underestimation of $l_{n,m}$ can be considerable and 
2) $l_{n,m}$ seems to perform 
better in dense graphs.

\section{A conjectural tight lower bound}

In this part, we conjecture an improved lower bound based on experimental results. First we focus on graphs 
$G=(V,E)$ with a universal vertex $u \in V$. Recall from 
Theorem~\ref{teo:universal_vertex} that the star spanning 
tree $T_s$ is a solution of the MSTCI problem. 
Note that the sum of the degrees of the 
subgraph $G' = G - \{u\}$ verifies 

$$\sum_{v \in G'} d(v) = 2 \nu.$$

 The following equivalent definitions refer to those graphs with 
 a universal vertex that minimizes the intersection number 
 for a fixed number of vertices and edges. The interpretation 
 is simply the equidistribution of the total degree among 
 the remaining $n-1$ vertices.

\begin{definition}
	Let $G=(V,E)$ be a graph with a universal 
	vertex $u \in V$ and $2 \nu = q (n-1) + r$ the 
	integer division of $2 \nu$ and $n-1$. We shall say 
	that $G$ is \emph{$\nu$-regular} if the degree 
	of every node $v \in V - \{u\}$ is $q+1$  or $q+2$.
\end{definition} 

\begin{definition}
	Let $G=(V,E)$ be a graph with a universal 
	vertex $u \in V$ and $2 \nu = q (n-1) + r$ the 
	integer division of $2 \nu$ and $n-1$. We shall say 
	that $G$ is \emph{$\nu$-regular}  if it has exactly $n-1-r$ nodes with degree $q+1$ and 
	$r$ nodes with degree $q+2$.
\end{definition} 

The equivalence is based on 
the uniqueness of $q$ and $r$. Note that for every $|V| = n$ and $|E| = m$ there are $\nu$-regular graphs, the proof can be expressed by induction on 
$\nu$.

\bigskip

The following lemma shows that 
$\nu$-regular graphs minimize the intersection 
number of graphs with a universal vertex.

\begin{lemma}\label{lemma:mu_equidistributed}
	If $G=(V_G,E_G)$ is a $\nu-regular$ graph and 
	$H=(V_H, E_H)$ is a non $\nu-regular$ graph  
	with a universal vertex such that $|V_G| = |V_H| = n$ and 
	$|E_G| = |E_H| = m$, then 

$$\cap(G) < \cap(H).$$
\end{lemma}

\begin{proof}
Suppose on the contrary that $\cap(H) \leq \cap(G)$. 
Let $u_H$ be a universal vertex of $H$. And -without loss of 
generality- let $H$ have a minimum intersection number; more 
precisely let $J = (V_{J}, E_{J})$ be a graph with a 
universal vertex such that $|V_{J}| = n$ and $|E_{J}| = m$. Then the following holds,

$$\cap(H) \leq \cap(J).$$

By definition as $H$ is not $\nu$-regular then a 
node with maximum degree and a node with minimum degree 
$v_{max}, v_{min} \in V_H-\{u\}$ satisfy that 
$d_H(v_{max}) - d_H(v_{min}) >= 2$. Let 
$(v_{max},w) \in E_H$ be an edge such that $w \in V_H$ is 
not a neighbor of $v_{min}$. Consider the graph 
$H' = (V_H, E_H - \{(v_{max},w)\} \cup \{(v_{min},w)\})$. 
According to the hypothesis on $H$ and the formula presented
in Lemma~\ref{lemma:star_formula} the following holds,

$$\cap(H) = \sum_{v \in V_H-\{u\}} \binom{d_H(v) - 1}{2} \leq \sum_{v \in V_H-\{u\}} \binom{d_{H'}(v) - 1}{2} = \cap(H').$$

As the degrees of all the nodes except for $v_{max}$ and 
$v_{min}$ coincide in $H$ and $H'$, the above inequality  
implies that,

$$\binom{d_H(v_{max}) - 1}{2} + \binom{d_H(v_{min}) - 1}{2} \leq 
\binom{d_{H'}(v_{max}) - 1}{2} + \binom{d_{H'}(v_{min}) - 1}{2}.$$

Let $d_H(v_{min}) - 1 = k$ and $d_H(v_{max}) - 1 = k+t$, 
rewriting the inequality we have,

$$\binom{k+t}{2} + \binom{k}{2} \leq 
\binom{k + t - 1}{2} + \binom{k + 1}{2}.$$

Expanding the binomials,

$$\frac{1}{2} [(k+t) (k+t-1) + k (k-1)] \leq 
\frac{1}{2} [(k+t-1) (k+t-2) + (k+1) k],$$

which implies,

$$(k+t-1) - k \leq 0 \implies t \leq 1 \implies d_H(v_{max}) - d_H(v_{min}) \leq 1.$$

This contradicts the hypothesis on $H$ and consequently proves the lemma.
\end{proof}

Considering the general case of connected graphs 
-in an attempt to improve $l_{n,m}$-, we decided to 
analyze those graphs that minimize the intersection 
number for each fixed $n$ and $m$. The experiment 
consisted of exhaustively checking the set of $7$ 
and $8$-node connected graphs. The interesting 
result was that in this general setting, 
$\nu$-regular graphs also minimize the 
intersection number; 
although not exclusively: in some cases there are 
other graphs -without a universal vertex- that 
achieve the minimum intersection number. We 
validated this fact by considering a sample of 1000 
randomly generated (by a uniform distribution)  
$9$-node connected graphs. This evaluation resulted 
positive in all cases, and consequently enables 
to formulate the following conjecture on a firm basis,

\begin{conjecture}\label{conj:mu_equidistributed}
	Let $G$ be a $\nu-regular$ graph and $H$ be a connected graph such that 
	$|V(G)| = |V(H)|$ and $|E(G)| = |E(H)|$. Then 

$$\cap(G) \leq \cap(H).$$
\end{conjecture}

Based on the intersection number formula of Lemma~\ref{lemma:star_formula}, this conjecture implies the 
following improved -and tight- lower bound of the 
intersection number.

\begin{corollary}\label{conj:lower_bound}
	If $G=(V,E)$ is a connected graph and 
	$2 \nu = q (n-1) + r$ the integer division of $2 \nu$ 
	and $n-1$, then 

$$\hat l_{n,m} := (n-1-r) \binom{q}{2} + r \binom{q+1}{2} = (n-1) \binom{q}{2} + q \ r\leq \cap(G).$$
\end{corollary}

\subsection{Comparison between $l_{n,m}$ and $\hat l_{n,m}$}













\begin{figure}[h]
	\centering
	\includegraphics[scale = 0.5]{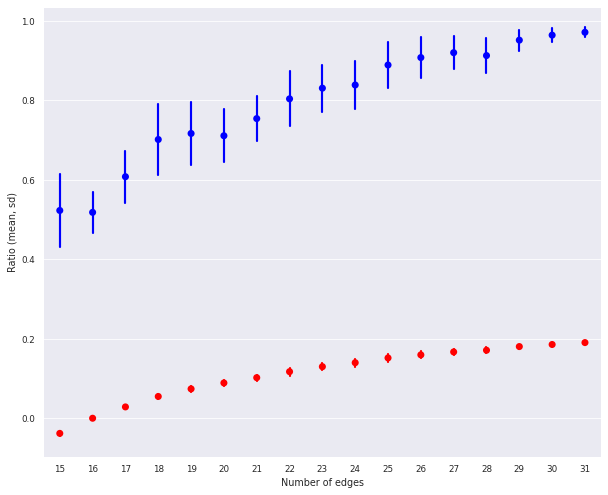}
	\caption{Comparison of mean and standard deviation of 
	$\frac{\hat l_{n,m}}{\cap(G)}$ (blue) and $\frac{l_{n,m}}{\cap(G)}$
	(red) over a uniformly distributed random sample of size 1000 of 
	$9$-node connected graphs.
  \label{fig:lower_bound_comparison}}
\end{figure}

Figure~\ref{fig:lower_bound_comparison} compares $\hat l_{n,m}$ and $l_{n,m}$ considering them as proportions of $\cap(G)$ on a random sample of 9-node connected graphs. This graphic suggests the following:

\begin{itemize}
	\item The performance improvement of $\hat l_{n,m}$ is much better than the one expressed in Lemma~\ref{lemma:bound_comparison}. 
	\item As the number of edges increases, $\hat l_{n,m}$ turns out be a tighter lower bound of $\cap(G)$.
\end{itemize}



\section{Conclusion} 

This article considers the Minimum Spanning Tree Cycle Intersection problem of arbitrary connected graphs. In this general setting, we focused on two lower bounds of the intersection number. The proof of the first lower bound ($l_{n,m}$) suggests two simple structural characteristics of the problem, a solution to it should be the ``best possible combination'' of the following conditions,

\begin{itemize}
	\item Tree-cycles should have short 
	length so that each cycle-edge belongs to 
	the least number of bonds.
	\item Cycle-edges should be 
	equidistributed among bonds.
\end{itemize}

The first condition resembles the Minimum Fundamental Cycle Basis problem~\cite{Kavitha:2009}, and the Low-stretch Spanning Tree problem~\cite{Alon1995,doi:10.1137/050641661}, whereas the second resembles some form of average-cut. These similarities could be useful in order to design algorithms. 

\bigskip

The second tight lower bound ($\hat l_{n,n}$) is presented in a conjectural form (Conjecture~\ref{conj:mu_equidistributed}). It is based on so-called $\nu$-regular graphs. These graphs verify the previous conditions and consequently reinforce their importance. This last bound shows that  graphs with a universal vertex play a fundamental role, both 
because they are well understood cases and they provide examples of minimum intersection number as a function of the number of nodes and edges. The bound can be expressed in a very general context as follows.

\medskip

\emph{Let G=(V,E) be a connected graph and T a spanning tree of G. Then 
}

$$(n-1) \binom{q}{2} + q \ r \leq \cap_G(T),$$

\noindent \emph{where $n = |V|$, $\nu$ is the cyclomatic number of G and $2 \nu = q (n-1) + r$ is the integer division of $2 \nu$ and $n-1$.}

\medskip

As mentioned in the introduction, a solution of the MSTCI problem can be useful for integrating discrete 1-forms on graphs. More specifically, fast linear solvers can be applied to sparse linear systems that involve cycle intersection matrices. We believe that this conjectured lower bound implies a good approximation of the sparsity of those matrices for arbitrary connected graphs. And this can help to decide if fast methods could be applied to linear systems of large graphs.

\bibliographystyle{amsplain}
\bibliography{mstci}

\end{document}